\documentclass{article}
\usepackage[english]{babel}
\usepackage{amssymb,upref,calc,url,amsmath,amscd,amsthm,verbatim} % calrsfs,cite
\usepackage[mathscr]{eucal}
\usepackage{graphicx}
\usepackage[all]{xy}
\usepackage[in]{fullpage}

\newtheorem{thm}{Theorem}
\newtheorem{cor}[thm]{Corollary}
\newtheorem{lem}[thm]{Lemma}
\newtheorem{ex}[thm]{Example}
\newtheorem{rem}[thm]{Remark}
\newtheorem{problem}[thm]{Problem}

\begin{document}

\title{Ordinary differential equations associated with the heat equation.}

\author{Victor M. Buchstaber, Elena Yu. Bunkova}

\date{}

\maketitle

\section*{Abstract.} \text{}

This paper is devoted to the one-dimensional heat equation 
and the non-linear ordinary differential equations associated to it.

We consider homogeneous polynomial dynamical systems 
in the n-dimensional space, $n = 0, 1, 2, \dots$.
For any such system our construction matches a non-linear ordinary differential equation.
 We describe the algorithm that brings the solution of such 
 an equation to a solution of the heat equation. 
The classical fundamental solution of the heat equation
corresponds to the case $n=0$ in terms of our construction. Solutions of the heat equation defined by the elliptic theta-function lead to the Chazy-3 equation and correspond to the case $n=2$.
  
The group $SL(2, \mathbb{C})$ acts on the space of solutions of the heat equation. We show this action for each $n \geqslant 0$ induces the action of $SL(2, \mathbb{C})$ on the space of solutions of the corresponding ordinary differential equations. In the case $n=2$ this leads to the well-known action of this group on the space of solutions of the Chazy-3 equation.
An explicit description of the family of ordinary differential equations arising in our approach is given.

\section{Introduction.} \text{}

The classical elliptic theta-function is a solution of the heat equation. Using the ansatz
\begin{equation}
\theta_1(z, 2 \pi i t) = e^{- {1 \over 2} h(t) z^2 + r(t)} \sigma(z, g_2(t), g_3(t))
\end{equation}
where $\sigma(z, g_2, g_3)$ is the Weierstrass sigma-function with parameters $g_2, g_3$, 
we get a polynomial dynamical system on the functions $g_2(t)$, $g_3(t)$ and $h(t)$.
It follows that $h(t)$ is a solution to the Chazy-3 equation and $g_2(t)$, $g_3(t)$ are differential polynomials of $h(t)$.
Therefore we get a family of solutions to the heat equation, parametrized by the initial data of the Cauchy problem for the Chazy-3 equation (see details in \cite{Trudy}).

This paper is devoted to the construction and study of a sequence of families of solutions $\psi(z, t; n)$, $n = 0, 1, \dots$ to the heat equation. For each $n$ we get a family of solutions of the heat equation parametrized by the parameters of a polynomial dynamical system and the initial data of the Cauchy problem of some ordinary differential equation.
For $n=2$ we get a family which includes the one described above.

The content of the paper is the following:
In section 2 we collect the necessary facts on the classical one-dimensional heat equation, in particular describe the action of the group $SL(2, \mathbb{C})$ on the space of its solutions. In section 3 for the solutions of the heat equation we study the ansatz
\begin{equation} \label{2}
\psi(z, t) = e^{- {1 \over 2} h(t) z^2 + r(t)} \Phi(z, \textit{\textbf{x}}(t)),
\end{equation} 
where the coefficients of $\Phi(z, \textit{\textbf{x}}(t))$ as a series of $z$ are polynomials of $\textit{\textbf{x}}(t) = (x_2(t), x_3(t), \dots, x_n(t))$. This ansatz is closed under the action of the group $SL(2, \mathbb{C})$. We present a general construction 
 that reduces the problem of solution of the heat equation in this ansatz to the solution of a homogeneous polynomial dynamical system. In section 4 we reduce homogeneous polynomial dynamical systems of the previous section to ordinary differential equations of a special type.  Therefore for any $n$ we get a family of solutions to the heat equation of the form \eqref{2}. Each of this solutions  is defined by a finite-dimensional numerical vector. In section 5 we study the ordinary differential equations obtained in the previous section and their relations for different $n$. In section 6 we study the recursions for the coefficients of the function $\Phi(z, \textit{\textbf{x}}(t))$. In section 7 we obtain rational solutions for a series of the ordinary differential equations obtained. In the last section the special cases $n=0,1,2,3,4$ are considered.
We pay special attention to differential equations with the Painlev\'e property.

We express our gratitude to R. Kont and J. C. Zambrini for the useful discussions of results of this paper. The work is supported by the RFFI grants 11-01-00197-a, 11-01-12067-ofi-m-2011, RF Government grant №2010-220-01-077, ag. 11.G34.31.0005.

\section{The general properties of the heat equation.} \text{}

The heat equation
\begin{equation} \label{heat}
{\partial \psi \over \partial t} = {1 \over 2} {\partial^2 \psi \over \partial z^2}
\end{equation}
is linear with respect to $\psi$ and invariant with respect to shifts of arguments:\\
--- for solutions $\psi_1$ and $\psi_2$ the function $\psi_1 + \psi_2$ is a solution,\\
--- for a solution $\psi(z,t)$ the function $\psi(z + z_0,t + t_0)$ is a solution for any constant $z_0$ and $t_0$, 
and thus ${\partial \over \partial z} \psi(z,t)$ and ${\partial \over \partial t} \psi(z,t)$ are solutions.

Therefore without loose of generality we may consider only solutions $\psi(z,t)$ of the heat equation
that are even or odd functions of $z$ regular at $(z,t) = (0,0)$. Any other solution will be a sum of such solutions up to a shift of arguments.

Consider the problem of finding solutions to \eqref{heat} with the initial conditions $\psi(0,t) = \psi_0(t)$, $\psi'(0,t) = \psi_1(t)$.

For a solution of the form 
\begin{equation} \label{psi1}
\psi(z,t) = \sum_{k=0}^\infty \psi_k(t) {z^k \over k!} 
\end{equation}
equation \eqref{heat} takes the form
\begin{equation} \label{psi}
\sum_{k=0}^\infty \psi_k'(t) {z^k \over k!}  = {1 \over 2} \sum_{k=0}^\infty \psi_{k+2}(t) {z^k \over k!} 
\end{equation}
and therefore is equivalent to the system $\psi_{k+2}(t)  = 2 \psi_k'(t)$, $k = 0, 1, \dots$. Thus any solution of the form \eqref{psi1} is a sum of an even solution with the initial conditions $\psi(0,t) = \psi_0(t)$, $\psi'(0,t) = 0$ and an odd solution with the initial conditions $\psi(0,t) = 0$, $\psi'(0,t) = \psi_1(t)$.

Let us consider a function of the form
\begin{equation} \label{formp}
\psi(z,t) = e^{- {1 \over 2} h(t) z^2} \phi(z,t),
\end{equation}
where $\phi(z,t)$ is a polynomial of $z$. 

\begin{thm} \label{thm0}
The function $\psi(z,t)$ of the form \eqref{formp} is a solution to the heat equation \eqref{heat} 
if and only if it is a linear combination with constant coefficients of the function
\begin{equation} \label{fexp}
{1 \over \sqrt{t-c}} \exp\left( {- z^2 \over 2 (t-c)}\right)
\end{equation}
and its derivatives with respect to $z$. Here $c$ is a constant.
\end{thm} 

\begin{proof}
If $\phi(z,t)$ is a polynomial of degree $n$, set $\phi(z,t) = \sum_{k=0}^n \phi_k(t) z^k$, $\phi_n(t) \not\equiv 0$. Substituting \eqref{formp} into \eqref{heat} and dividing by $e^{- {1 \over 2} h(t) z^2}$, we get an equality between two polynomials of $z$. The one in the left hand side has the coefficient $- {1 \over 2} h' \phi_n$ at $z^{n+2}$, and the one at the right hand side has the coefficient ${1 \over 2} h^2 \phi_n$ at $z^{n+2}$. Thus 
\[
h(t) = {1 \over (t - c)}
\]
for some constant $c$.
Now the left hand side has the coefficient $\phi_n'(t)$ at $z^{n}$, and the right hand side has the coefficient $- (n + {1 \over 2}) h(t) \phi_n(t)$ at $z^{n}$. 
Thus $\phi_n(t) = {c_n \over (t - c)^{n + {1 \over 2}}}$ for some constant $c_n \ne 0$.
A solution with this $\phi_n(t)$ is $\psi_n(t) = (- 1)^{n} c_n {\partial^n \over \partial z^n} \left( {1 \over \sqrt{t - c}} exp\left({- z^2 \over 2 (t - c)}\right) \right)$.
The other coefficients give the relations 
\[
\phi_{k+2} = {2 \over (k+2) (k+1)} \left( \phi_k'(t) + {(2 k +1) \over 2} h(t) \phi_k(t) \right), \quad k = 0, 1, 2, \dots, n-1
\]
which for fixed $\phi_{k+2}$ define $\phi_{k}$ uniquely up to the addition of a solution to $\phi_k'(t) = - {(2 k +1)\over 2} h(t) \phi_k(t)$, that is $\psi_k(t)$. Thus for any polynomial $\phi(z,t)$ the corresponding solution is the sum of $\psi_k(t)$ with some coefficients, and $\psi_k(t)$ themselves up to constants are derivatives of \eqref{fexp}.
\end{proof}

The \emph{Hermite polynomials} $He_k(x)$ 
are defined by the relation
\[
{d^k \over d x^k} \exp \left(- x^2 \over 2\right) = (-1)^k \exp \left( - x^2 \over 2 \right) He_k(x).
\]  

\begin{cor}
The function $\psi(z,t)$ is a solution of the form \eqref{formp} of the equation \eqref{heat}
if and only if $h(t) = {1 \over (t - c)}$ for some constant $c$ and $\phi(z,t)$ is a linear combination
with constant coefficients of the polynomials with respect to $z$
\[
{1 \over \sqrt{t - c}} He_k({z \over \sqrt{t - c}}).
\]
\end{cor}

In the focus of our interest is the construction of a sequence 
of solutions of the heat equation starting with \eqref{fexp}. 
Each of this solutions generates by differentiation 
with respect to $z$ a new sequence of solutions. 

In the case $t \in \mathbb{R}$, $z \in \mathbb{R}$ set 
\begin{equation} \label{I}
I(t) = \int_{- \infty}^{\infty} \psi(z,t) dz.
\end{equation}

Let $\psi(z,t)$ be a solution of the heat equation. Then if 
\[
\left. {\partial \psi \over \partial z} \right|_{- \infty}^{\infty} \equiv 0 \text{ for } t_0 < t < t_1,
\]
then $I(t) \equiv const$ on the interval $(t_0, t_1)$.

\begin{cor}
In the notions of theorem \ref{thm0} the conservation law for \eqref{I} holds:
\[
I(t) \equiv const \text{ for } t \in (c, \infty).
\]
\end{cor}

\subsection{Symmetry groups.} \text{}

The detailed description of the Lie algebra of heat equation transform group and one-parametric subgroups corresponding to the basis in this Lie algebra are given in \cite{Olver}. We will need the action of the group $SL(2,\mathbb{C})$ on the space of solutions of the heat equation.

Let $M = \begin{pmatrix} a & b \\ c & d \end{pmatrix} \in SL(2,\mathbb{C})$ and
\begin{equation} \label{shift}
\Gamma(M) \psi(z,t) = {1 \over \sqrt{c t + d}} \exp{\left( {- c z^2 \over 2 (c t + d)}\right)} \psi\left({z \over c t + d},{a t + b \over c t + d}\right).
\end{equation}

\begin{lem} \label{Le}
If $\psi(z,t)$ is a solution to the heat equation \eqref{heat}, then so is $\Gamma(M) \psi(z,t)$.
\end{lem}

The proof is a straightforward substitution of \eqref{shift} into \eqref{heat} using $a d - b c = 1$.

Note $\Gamma(M) 1$ is the classical solution \eqref{fexp} of the heat equation.

\begin{lem} 
The group $SL(2, \mathbb{C})$ acts on the space of solutions of the heat equation.
\end{lem}

\begin{proof}
We have $\Gamma(M_2) \Gamma(M_1) 1 = \Gamma \left( M_1 M_2 \right)1$. Thus $\Gamma(M)$ induces a representation of $SL(2,\mathbb{C})$ on the space of solutions of the heat equation.
\end{proof}

\section{Heat equations and dynamical systems.} \text{}

Any even (odd) regular at $z = 0$ 
nonvanishing function $\psi(z,t)$ may be presented in the form 
\begin{equation} \label{form}
\psi(z,t) = e^{- {1 \over 2} h(t) z^2 + r(t)} z^{2 s} \phi(z,t),
\end{equation}
where $s$ is an non-negative integer and 
the function $\phi(z,t)$ in the vicinity of $z=0$ is given by the series
\begin{equation} \label{e3}
\phi(z,t) = z^\delta + \sum_{k \geqslant 2} \phi_k(t) {z^{2k+\delta} \over (2k+\delta)!},
\end{equation}
where $\delta = 0$ (correspondingly, $\delta = 1$). Note that this representation is unique, thus for any function $\psi(z,t)$ we may speak of the corresponding functions $h(t)$, $r(t)$, $\phi(z,t)$. It follows from \eqref{psi} that if $\psi$ solves the heat equation then $s=0$.
Further we assume $\delta = 0$ or $\delta = 1$.

Consider a function $\psi(z,t)$ in the form \eqref{form} with $s = 0$ and a function $h(t)$. 
Denote by $\gamma_1(M) h(t)$ and $\gamma_2(M) \exp(r(t))$ the corresponding functions in the representation in the form \eqref{form} of the function $\Gamma(M) \psi(z,t)$.

\begin{lem} \label{Let}
\[
\gamma_1(M) h(t) = {1 \over (c t + d)^2} h\left({a t + b \over c t + d}\right) + {c \over (c t + d)}, \qquad \gamma_2(M) \exp(r(t)) = {1 \over (c t + d)^{\delta} \sqrt{c t + d}} e^{r\left({a t + b \over c t + d}\right)}
\]
\end{lem}
\begin{proof}
For $\psi(z,t)$ of the form \eqref{form} with $s = 0$ we have
\[
\Gamma(M) \psi(z,t) = {1 \over \sqrt{c t + d}} \exp{\left( - {1 \over 2} \left({1 \over (c t + d)^2} h\left({a t + b \over c t + d}\right) + {c \over (c t + d)} \right) z^2 + r\left({a t + b \over c t + d}\right) \right)} \phi\left({z \over c t + d},{a t + b \over c t + d}\right)
\]
and $\phi({z \over c t + d},{a t + b \over c t + d}) = {1 \over (c t + d)^{\delta}} \hat{\phi}(z,t)$ with $\hat{\phi}(z,t) = z^\delta + \sum_{k \geqslant 2} \hat{\phi}_k(t) {z^{2k+\delta} \over (2k+\delta)!}$ for $\hat{\phi}_k(t) = {1 \over (c t + d)^{2 k}} \phi_k({a t + b \over c t + d})$.
We see $\hat{\phi}(z,t)$ has the from \eqref{e3},
thus the action on $e^{r(t)}$ and on $h(t)$ induced from the action of $\Gamma(M)$ on $\psi(z,t)$ in \eqref{form} is as in the lemma.
\end{proof}

\subsection{General construction.} \label{secgc} \text{}

\begin{thm} \label{thm1}
The function $\psi(z,t)$ of the form \eqref{form}
is a solution to the heat equation \eqref{heat} 
if and only if $s =0$, $r' = - (\delta + {1 \over 2}) h$, 
and the function $\phi(z,t)$ is a solution to
\begin{equation} \label{e4}
{\partial \over \partial t} \phi = \mathcal{H}_2 \phi - h \mathcal{H}_0 \phi,
\end{equation}
where
\[
\mathcal{H}_2 \phi = \left({1 \over 2} {\partial^2 \over \partial z^2} + u z^2 \right) \phi,
\qquad
u =  {1 \over 2} \left( h' + h^2 \right),
\qquad
\mathcal{H}_0 \phi = \left( z {\partial \over \partial z} - \delta \right) \phi.
\]
\end{thm}

\begin{proof} Substituting \eqref{form} and \eqref{e3} into \eqref{heat} and presenting the right and left parts of this equation as series of $z$,
up to elements of order $2 s + \delta$ we get
\[
(2 s + \delta) (2 s + \delta - 1) e^{r(t)} z^{2 s + \delta - 2} + (z^{2 s + \delta}) = 0.
\]
Thus $s = 0$ and up to elements of order $\delta + 2$ we get
\[
( (\delta + {1 \over 2}) h(t) + r'(t)) e^{r(t)} z^{\delta} + (z^{\delta + 2}) = 0.
\]
Thus $r'(t) = - (\delta + {1 \over 2}) h(t)$, and the last statement of the theorem is obtained directly from substituting \eqref{form} into \eqref{heat} and dividing by $e^{- {1 \over 2} h(t) z^2 + r(t)}$. 
\end{proof}

Set $\textbf{\textit{x}} = (x_2, \dots, x_{n+1})$.
We will assume that the variables $x_q$ have the grading $- 4 q$, and all constants (that is the parameters not depending on $t$, $z$ nor $x$) have the grading $0$ if not stated otherwise.
Set
\begin{equation} \label{e5}
\Phi(z; \textbf{\textit{x}}) = z^\delta + \sum_{k \geqslant 2} \Phi_k(\textbf{\textit{x}}) {z^{2 k+ \delta} \over (2 k + \delta)!},
\end{equation}
where $z$ is a variable of grading $2$ independent of $\textbf{\textit{x}}$, 
and $\Phi_k(\textbf{\textit{x}})$ are homogeneous polynomials of degree $-4 k$.
For example, we have $\Phi_2(\textbf{\textit{x}}) = c_2 x_2$ for some constant $c_2$ and $\Phi_3(\textbf{\textit{x}}) = c_3 x_3$ for some constant $c_3$.
Therefore $\Phi(z; \textbf{\textit{x}})$ is a homogeneous function of degree $2 \delta$. 
Let us remark that functions homogeneous with respect to some degree are often called quasi-homogeneous. We will not make any difference and further on call all this functions homogeneous.

The homogeneity of $\Phi(z; \textbf{\textit{x}})$ implies 
\[
\mathcal{H}_0 \Phi = L_0 \Phi, 
\quad
\text{where}
\quad
L_0 = \sum_{k=2}^{n+1} 2 k x_k {\partial \over \partial x_k}.
\]

Let us consider in $\mathbb{C}^n$ with coordinates $(x_2, \dots, x_{n+1})$ 
a set of homogeneous polynomials $ p_{q}(\textbf{\textit{x}})$ with $\deg p_q = - 4 q$, $q = 3, \dots, n+2$, 
and introduce a homogeneous polynomial dynamical system
\begin{equation} \label{e11}
{d x_k \over d \tau} = p_{k+1}(\textbf{\textit{x}}), \quad k = 2, \dots, n+1, \quad \deg \tau = 4.
\end{equation}
Set
\[
L_2 \Phi = \sum p_{k+1} (x_2, \dots, x_{n+1}) {\partial \over \partial x_k} \Phi.
\]

We will need the following result: 

\begin{thm} \label{t2}
The function $\Phi(z; \textbf{\textit{x}})$
gives a solution to the equation
\begin{equation} \label{e6}
\mathcal{H}_2 \Phi = L_2 \Phi
\end{equation}
if and only if 
\begin{equation} \label{e7}
\Phi_2 = - 4 (1 + 2 \delta) u, \quad
\Phi_{3} = 2 {d \over d \tau} \Phi_2, \quad
\Phi_{k+1} = 2 {d \over d \tau} \Phi_k + {(2 k + \delta - 1) (2 k + \delta) \over 2 (1 + 2 \delta)} \Phi_2 \Phi_{k-1}, \quad k>2.
\end{equation}
\end{thm}

The proof follows from direct substitution of \eqref{e5} into \eqref{e6}.

\begin{cor}
Under the conditions of this theorem for a constant $c$ such that $\Phi_2(\textbf{\textit{x}}) = c x_2$ we have $u = - {c \over 4 (1 + 2 \delta)} x_2$.
\end{cor}

Recall a function $f(\textbf{\textit{x}})$ is called non-degenerate if for a vector-function $\phi(\textbf{\textit{x}})$ 
the equality $ \langle \phi(\textbf{\textit{x}}), grad \; f(\textbf{\textit{x}}) \rangle \equiv 0 $ implies $\phi(\textbf{\textit{x}}) \equiv 0$.

\begin{thm}
Let $\Phi(z; \textbf{\textit{x}})$ be a non-degenerate function of the form \eqref{e5} 
with coefficients satisfying \eqref{e7}.
The function   
\begin{equation} \label{e8}
\phi(z,t) = \Phi(z; x_2(t), \dots, x_{n+1}(t))
\end{equation}
is a solution of the equation \eqref{e4} if and only if the vector-function $\textbf{\textit{x}}(t) = (x_2(t), \dots, x_{n+1}(t))$ 
satisfies the dynamical system
\[
{d \over d t} x_k(t) = p_{k+1}(\textbf{\textit{x}}(t)) - 2 k h(t) x_k(t), \quad k = 2, \dots, n+1.
\]
\end{thm}

\begin{proof}
With the substitution \eqref{e8} equation \eqref{e4} becomes
\[
\sum_{k=2}^{n+1}  \left({\partial x_k \over \partial t} - p_{k+1}(x_2, \dots, x_{n+1}) + 2 k h x_k \right) {\partial \Phi \over \partial x_k} = 0.
\]
The condition that $\Phi$ is non-degenerate finishes the proof.
\end{proof}

Let us summarize the results of this section in a single theorem.
Let $p_{k}(\textbf{\textit{x}})$, $k = 3, \dots, n+2$, be homogeneous polynomials of $\textbf{\textit{x}} = (x_2, \dots, x_{n+1})$, $\deg x_q = - 4 q$, $\deg p_k = - 4 k$, $c$ -- some constant, $\Phi(z; \textbf{\textit{x}})$ -- a homogeneous function defined by \eqref{e5}.

\begin{thm}[Key theorem] \label{t5}
Among the following conditions each two imply the third one:

1) 
The function
\begin{equation} \label{sol}
\psi(z,t) = e^{- {1 \over 2} h(t) z^2 + r(t)} \Phi(z; \textbf{\textit{x}}(t))
\end{equation}
solves the heat equation.

2) The coefficients of $\Phi(z; \textbf{\textit{x}})$ are defined by the recursion 
\[
\Phi_{q} = 2 \sum_{k=2}^{n+1} p_{k+1}(\textbf{\textit{x}}) {\partial \over \partial x_k} \Phi_{q-1} + {(2 q + \delta - 3) (2 q + \delta - 2) \over 2 (1 + 2 \delta)}  \Phi_2 \Phi_{q-2}, \qquad q = 4, 5, 6, \dots,
\]
with the initial conditions 
$\Phi_2(\textbf{\textit{x}}) = c x_2$, $\Phi_{3} = 2 c p_{3}(\textbf{\textit{x}})$.

3) The set of functions $(r(t), h(t), \textbf{\textit{x}}(t))$ solves the homogeneous polynomial dynamical system
\begin{equation} \label{e9}
{d \over d t} r = - (\delta + {1 \over 2}) h, \quad
{d \over d t} h = - h^2 - {c \over 2 (1 + 2 \delta)} x_2,\quad
{d \over d t} x_k = p_{k+1}(\textbf{\textit{x}}) - 2 k h x_k, \quad k = 2, \dots, n+1.
\end{equation}
\end{thm}

The action $\Gamma$ of the group $SL(2, \mathbb{C})$ on solutions of the heat equation is described above.

\begin{cor} \label{ts}
Let $\psi(z,t)$ be a solution to the heat equation of the form \eqref{sol}. Then 
\[
\Gamma(M) \psi(z,t) = e^{- {1 \over 2} \widehat{h}(t) z^2 + \widehat{r}(t)} \Phi(z; \widehat{\textbf{\textit{x}}}(t))
\]
is a solution to the heat equation with $\widehat{h}(t) = \gamma_1(M) h(t)$, $e^{\widehat{r}(t)} = \gamma_2(M)e^{r(t)}$ and $\widehat{x}_k(t) = {x_k(t) \over (c t + d)^{2k}}$.
\end{cor}

\section{Dynamical systems 
and ordinary differential equations.}  \label{c3} \text{}

For a set of homogeneous polynomials $p_k(\textbf{\textit{x}})$ with $\deg p_k = - 4 k$, $k = 3, \dots, n+2$,
the grading condition implies that for any $k \leqslant n$
 the polynomial $p_k(\textbf{\textit{x}})$ does not depend on $x_{k+1}, \dots, x_{n+1}$. 
We may further denote the polynomials $p_k(x_2, \dots, x_{n+1})$
not depending on $x_s, \dots, x_{n+1}$ by $p_k(x_2, \dots, x_{s-1})$.

\begin{lem}
The group of polynomial transforms of the form
\begin{equation} \label{e12}
X_2 = c_2 x_2, \quad X_k = c_k x_k + q_k(x_2, \dots, x_{k-1}), \quad k = 3, \dots, n+1,
\end{equation}
where $c_k \ne 0$ are constants and $q_k(x_2, \dots, x_k)$ are homogeneous polynomials, $\deg q_k = - 4 k$, acts on the space of homogeneous polynomial dynamical systems of the from \eqref{e11}. 
This action brings the system \eqref{e11} to
\[
{d \over d \tau} X_k = P_{k+1}(X_2, \dots, X_{n+1}), \quad k = 3, \dots, n+2\]
where
\begin{equation} \label{P}
P_{k+1}(X_2, \dots, X_{n+1}) = c_k p_{k+1}(x_2, \dots, x_{n+1}) - {d \over d \tau} q_k(x_2, \dots, x_{k-1}).
\end{equation}
\end{lem}

\begin{rem}
${d \over d \tau} q_k(x_2, \dots, x_{k-1})$ is a short form for $\sum_q {\partial \over \partial x_q} q_k(x_2, \dots, x_{k-1}) p_{q+1}(x_2, \dots, x_{n+1})$.
\end{rem}

The proof of the lemma is a straightforward substitution.

\begin{lem} \label{l6} 
The group of polynomial transforms \eqref{e12} 
acts on dynamical systems of the form \eqref{e9}
 and brings a solution 
$(r(t), h(t), \textbf{\textit{x}}(t))$ to 
a solution $(r(t), h(t), X_2(t), \dots, X_{n+1}(t))$ of the dynamical system
\[
{d \over d t} r =  - (\delta + {1 \over 2}) h, 
\quad
{d \over d t} h = - h^2 - {c \over 2 c_2 (1 + 2 \delta)} X_2,
\quad
{d \over d t} X_k = P_{k+1}(X_2, \dots, X_{n+1}) - 2 k h X_k, \quad k = 2, \dots, n+1.
\]
\end{lem}

\begin{proof}
Each transform \eqref{e12} may be presented as a composition of transforms of the form 
\begin{equation} \label{eq1}
X_k = c_k x_k, \quad c_k \ne 0, \quad k = 2, \dots, n+1,
\end{equation}
and
\begin{equation} \label{eq2}
X_k = x_k, \quad k \ne s, \quad X_s = x_s + a_0 x_2^{j_2} \dots x_{s-2}^{j_{s-2}} \quad \text{with} \quad \sum_{m=2}^{s-2} m j_m = s, \quad a_0 = const.
\end{equation}
Now \eqref{eq1} brings \eqref{e9} to 
\[
{d \over d t} r = - (\delta + {1 \over 2}) h, \quad
{d \over d t} h = - h^2 - {c \over 2 (1 + 2 \delta) c_2} X_2, \quad
{d \over d t} X_k = c_k p_{k+1}({X_2 \over c_2}, \dots, {X_{n+1} \over c_{n+1}}) - 2 k h X_k, \quad k = 2, \dots, n+1.
\]
and \eqref{eq2} brings \eqref{e9} to 
\[
{d \over d t} r = - (\delta + {1 \over 2}) h, \quad
{d \over d t} h = - h^2 - {c \over 2 (1 + 2 \delta)} X_2, \quad
{d \over d t} X_k = P_{k+1}(X_2, \dots, X_{n+1}) - 2 k h X_k, \quad k = 2, \dots, n+1, 
\]
where $P_{k+1}(X_2, \dots, X_{n+1})$ is defined by \eqref{P}. The last statement is true because from \eqref{e9} it follows that $F = x_2^{j_2} \dots x_{s-2}^{j_{s-2}}$, $\sum_{m=2}^{s-2} m j_m = s$, solves the equation ${d \over d t} F = \sum p_{k+1}(\textit{\textbf{x}}) {\partial \over \partial x_k} F + 2 s h F$.
\end{proof}

\begin{lem} \label{l11}
If $p_s(\textbf{\textit{x}})$ does not depend on $x_s$ for some $s$, 
then the set $(r(t), h(t), x_2(t), \dots, x_{s-1}(t))$ solves the system
\begin{equation} \label{e13}
{d \over d t} r = - (\delta + {1 \over 2}) h, \quad {d \over d t} h = - h^2 - {c \over 2 (1 + 2 \delta)} x_2, 
\quad
{d \over d t} x_k = p_{k+1}(x_2, \dots, x_{s-1}) - 2 k h x_k, \quad k = 2, \dots, s-1.
\end{equation}
This solution along with the set of polynomials $p_{k}(x_2, \dots, x_{s-1})$, $k = 3, \dots, s$,
determines a solution of the heat equation which coincides with the solution \eqref{sol} described above.
\end{lem}

Thus, in the conditions of the lemma, the solution of the heat equation does not depend on the polynomials $p_{s+1}, \dots, p_{n+2}$.

\begin{proof}
If $p_s(x_2, \dots, x_{n+1})$ does not depend on $x_s$, then all $p_k(x_2, \dots, x_{n+1})$, $k = 3, \dots, s$ do not depend on $x_s, \dots, x_{n+1}$ because of the grading, thus \eqref{e13} is a subsystem of \eqref{e9} with only $h(t), r(t), x_2(t), \dots, x_{s-1}(t)$ as variables. Therefore $\Phi_k$ are functions of $(h(t), r(t), x_2(t), \dots, x_{s-1}(t))$ because $\Phi_2$ and $p_k$, $k \leq s$, are such functions, and $\Phi_k$ are defined recurrently as such functions, and thus the solution \eqref{sol} of the heat equation does not depend on $x_s(t), \dots, x_{n+1}(t)$ and is defined only by a solution to \eqref{e13}.
\end{proof}

A polynomial dynamical system \eqref{e11} is called \emph{reduced} if it is defined by the set
\begin{equation} \label{red}
p_k(\textbf{\textit{x}}) = x_k, \quad k = 3, \dots, n+1,
\quad
p_{n+2}(\textbf{\textit{x}}) = P_n(\textbf{\textit{x}}),
\end{equation}
where $P_n(\textbf{\textit{x}})$ is a homogeneous polynomial of degree $ - 4 (n+2)$.

\begin{lem}
Each solution of the form \eqref{sol} of the heat equation may be obtained by our construction using a reduced dynamical system \eqref{red}
and a solution to the system
\begin{multline} \label{e14}
{d \over d t} r = - (\delta + {1 \over 2}) h, \quad
{d \over d t} h = - h^2 + x_2, 
 \quad
{d \over d t} x_k = x_{k+1} - 2 k h x_k, \quad k = 2, \dots, n,
\\
{d \over d t} x_{n+1} = P_n(x_2, \dots, x_{n+1}) - 2 (n+1) h x_{n+1}.
\end{multline}
\end{lem}

\begin{proof}
Let us start from a set of homogeneous polynomials $p_{k}(x_2, \dots, x_{n+1})$, $k = 3, \dots, n+2$.
According to lemma \ref{l11} we may assume that $p_{k}(x_2, \dots, x_{n+1})$ depend substantially on $x_{k}$ for $k = 3, \dots, n+1$, thus 
in \eqref{e12} we may take $X_2 = - {c \over 2 (1 + 2 \delta)} x_2$ and recurrently
\[
X_k = c_{k-1} p_k(x_2, \dots, x_{n+1}) + {d \over d \tau} q_{k-1}(x_2, \dots, x_{k-1}) \quad \text{for} \quad k = 3, \dots, n+1,
\]
$c_{k-1}$ and $q_{k-1}$ being defined by \eqref{e12}.
The polynomial system will take the form we need.
\end{proof}

The system \eqref{e14} is defined by a number $n$ and 
the polynomial $P_n(\textbf{\textit{x}})$. 
In this system
$r$ is defined as a function of $h(t)$ up to a constant of integration $r_0$ by
 $r' = - (\delta + {1 \over 2}) h$ and $x_2, \dots, x_{n+1}$ are defined by
\begin{equation} \label{subst}
x_2 = h' + h^2, \quad
x_{k} = x_{k-1}' + 2 (k-1) h x_{k-1}, \quad k = 3, \dots, n+1.
\end{equation}
Substituting $x_k$ as functions of $h(t)$ into 
\[
{d \over d t} x_{n+1} = P_n(x_2, \dots, x_{n+1}) - 2 (n+1) h x_{n+1},
\]
we get an ordinary differential equation 
\begin{equation} \label{nomer}
\mathcal{D}_{P_n,n+1}(h) = 0
\end{equation}
of order $n+1$ on $h(t)$. 
This equation is homogeneous with respect to the grading $\deg h = - 4$, $\deg t = 4$. 

Summarizing the results of this section, we obtain:

\begin{thm} \label{t9}
Each solution of the heat equation of the form \eqref{sol} is defined by the set $(n, P_n, h, r_0)$, where

$n$ is a natural number, 

$P_n$ is a homogeneous polynomial $P_n(\textbf{\textit{x}})$ of degree $- 4 (n+2)$, 

$h$ is a solution $h(t)$ of the equation $\mathcal{D}_{P_n,n+1}(h) = 0$, 

$r_0$ is a constant.
\end{thm}

\begin{cor}
Each solution of the heat equation of the form \eqref{sol} is defined by a finite-dimensional numerical vector $(n, \textbf{\textit{P}}_n, \textbf{\textit{h}}_n, r_0)$, where $ \textbf{\textit{P}}_n$ is the vector of coefficients of the polynomial $P_n(\textbf{\textit{x}})$ and $\textbf{\textit{h}}_n$ is the vector of initial data in the Cauchy problem for the equation \eqref{nomer}.
\end{cor}

\section{Ordinary differential equations, associated with the heat equation.} \label{c34} \text{}

Let $P_n(\textbf{y})$ be a polynomial of an $(n-1)$-dimensional argument, 
such that $P_n(\textbf{\textit{x}})$ for the graded vector $\textbf{\textit{x}} = (x_2, \dots, x_{n+1})$ 
is homogeneous of degree $-4(n+2)$.

Denote by $V_n$ the linear vector space of such polynomials. 

\begin{lem} \label{l1}
We have
\[
\dim V_n = p(n+2) - p(n+1) - 1,\]
where $p(n)$ is the number of partitions of $n$. 
\end{lem}

\begin{proof}
The dimension of the space of homogeneous polynomials of degree $-4(n+1)$ of the variables $x_1, \dots, x_{n+1}$ with $\deg x_k = - 4 k$ is $p(n+1)$. Thus $p(n+2) - p(n+1)$ is the dimension of the space of polynomials of $x_1, \dots, x_{n+2}$ of degree $- 4 (n+2)$ that can not be presented in the form $x_1 p$, where $p$ is a polynomial of degree $- 4 (n+1)$. Taking away the polynomials $p(x_1, \dots, x_{n+2}) = c x_{n+2}$ we get $V_n$.
\end{proof}

Consider the equation \eqref{nomer}.
In the case $P_n \equiv 0$ set $\mathcal{D}_{0,n+1}(h) = \mathcal{D}_{n+1}(h)$.

\begin{thm} \label{t21}
Set
\[
\mathcal{D}_1(h) = ({d \over d t} + h) h.
\]
Then for each $n \geqslant 2$ the formula holds
\begin{equation} \label{star}
\mathcal{D}_n(h) = ({d \over d t} + 2 n h) \mathcal{D}_{n-1}(h).
\end{equation}
\end{thm}

\begin{proof}
The equation $\mathcal{D}_{n+1} = 0$ is obtained from 
\[
{d x_{n+1} \over d t} = - 2 (n+1) h x_{n+1} 
\]
by the substitution \eqref{subst}. For $n=0$ there are no substitutions to make and the equation is $h' + h^2 = 0$.
Advancing from $\mathcal{D}_{n+1}$ to $\mathcal{D}_{n+2}$ is a substitution of $x_{n+2} = x_{n+1}' + 2 (n+1) h x_{n+1}$ into
\[
{d x_{n+2} \over d t} = - 2 (n+2) h x_{n+2}, 
\]
that is the action of $\left({d \over d t} + 2 (n+2) h \right)$ on $\mathcal{D}_{n+1}(h)$.
\end{proof}

\begin{cor}
For $n>1$ we have
\[
\mathcal{D}_{n}(h) = h^{(n)} + n (n+1) h^{(n-1)} h + \dots +  2^{n - 1} \cdot n! \; h^{n+1}.
\]
\end{cor}

\begin{thm}
For the $(n-1)$-dimensional vector-function
\[
\mathcal{D}(h) = (\mathcal{D}_1(h), \dots, \mathcal{D}_{n-1}(h))
\]
and the polynomial $P_n(\textbf{y}) \in V_n$ the formula holds
\[
\mathcal{D}_{P_n,n+1}(h) = \mathcal{D}_{n+1}(h) -  P_n(\mathcal{D}(h)).
\]
\end{thm}

The proof follows from \eqref{e14} using theorem \ref{t21}.

Set 
\[
d = \sum_{k=1}^\infty y_{k+1} {\partial \over \partial y_k}.
\]

\begin{thm}
For $P_n({\bf x}) \in V_n$ the formula holds
\[
\left({d \over dt} + 2 (n+2) h \right) \mathcal{D}_{P_n, n+1}(h) = \mathcal{D}_{P_{n+1}, n+2}(h),
\]
where 
\[P_{n+1}({\bf y}) = d P_{n}({\bf y}).\]
\end{thm}

\begin{proof}

For any monomial $c_{j_1, \dots, j_n} \mathcal{D}_1(h)^{j_1} \dots \mathcal{D}_n(h)^{j_n}$ of $P_n(\mathcal{D}(h)))$ we have
\begin{multline*}
\left({d \over dt} + 2 (n+2) h \right) c_{j_1, \dots, j_n} \mathcal{D}_1(h)^{j_1} \dots \mathcal{D}_n(h)^{j_n} =\\ = 
c_{j_1, \dots, j_n} \sum_k j_k \mathcal{D}_1(h)^{j_1} \dots \mathcal{D}_k(h)^{j_k - 1} \dots \mathcal{D}_n(h)^{j_n} \left({d \over dt} + 2 (k+1) h \right) \mathcal{D}_k(h) = \\ =  c_{j_1, \dots, j_n} \sum_k j_k \mathcal{D}_1(h)^{j_1} \dots \mathcal{D}_k(h)^{j_k - 1} \dots \mathcal{D}_n(h)^{j_n} \mathcal{D}_{k+1}(h).
\end{multline*}
Thus 
\[
\left({d \over dt} + 2 (n+2) h \right)  P_n(\mathcal{D}(h))) = (d P_n)(\mathcal{D}(h)).
\]
\end{proof}

For $n>1$ the polynomials $P_n$ may be nontrivial and
we get families of equations 
with coefficients of this polynomials as parameters.
For fixed values of this parameters we get special classes
 of differential equations.
Further we will consider in detail the equations 
$
\mathcal{D}_{P_n, n+1}(h) = 0,
$
where $n = 0,1,2,3,4$, and discuss the resulting classes of equations.

An ordinary differential equation is said to have 
the \emph{Painlev\'e property} if one can construct a single-valued
 general solution of this equation.
A characteristic of an equation with the Painlev\'e property is that its critical singular points do not depend of the initial data (see \cite{CM}).
Equations with the Painlev\'e property play an important role in modern physical problems. Such problems related to the Chazy-3 equation are described in \cite{BLP}, \cite{CAC}, \cite{D}. The action of the group $SL(2, \mathbb{C})$ on the space of its solutions is studied in \cite{CO}.
Further we will show that for some parameters
the equation $\mathcal{D}_{P_2,3}(h) = 0$
coincides with Chazy-3 and Chazy-12 equations,
thus has the Painlev\'e property.

According to corollary \ref{ts}, there is an action of $SL(2, \mathbb{C})$ on the space of solutions of each the differential equation on $h$ we cite as $\mathcal{D}_{P_n, n+1}(h) = 0$.

\section{Discrete differential equations associated to the heat equation.} \text{}

In our general construction (section \ref{secgc}) we construct homogeneous functions $\Phi(z; \textbf{\textit{x}})$ that satisfy the equation \eqref{e6} and thus provide solutions to the heat equation (see theorem \ref{t5}). In this section we return to the problem of construction of such functions.

The homogeneous function $\Phi(z; \textbf{\textit{x}})$ (see \eqref{e5}) can be presented in the form
\begin{equation} \label{disform}
\Phi(z; \textbf{\textit{x}}) = z^\delta + \sum_{||J|| \geqslant 4} a(J) \textbf{\textit{x}}^J {z^{||J||+ \delta} \over (||J|| + \delta)!},
\end{equation}
where as before $\textbf{\textit{x}} = (x_2, \dots, x_{n+1})$, $\deg x_k = - 4 k$, $\deg z = 2$, $J = (j_2, \dots, j_{n+1})$ is a multiindex, $a(J)$ are constant coefficients and $||J|| = \sum_{k=2}^{n+1} 2 k j_k$. That is in the notions above $\Phi_k(x) = \sum_{||J|| = 2 k} a(J) \textbf{\textit{x}}^J$.

For a reduced dynamical system \eqref{red} the equation \eqref{e6} has the explicit form
\begin{equation} \label{diseq}
\left({1 \over 2} {\partial^2 \over \partial z^2} + u z^2 \right) \Phi = P_n(\textbf{\textit{x}}) {\partial \over \partial x_{n+1}} \Phi + \sum_{k=2}^n x_{k+1} {\partial \over \partial x_k} \Phi.
\end{equation}
We have
$u = - {c \over 4 (1 + 2 \delta)} x_2$ where $c$ is a constant such that $\Phi_2 = c x_2$. 
Note that in the previous sections we used $c = - 2 (1 + 2 \delta)$.
Let us present the homogeneous polynomials $P_n(\textbf{\textit{x}})$ in the form $P_n(\textbf{\textit{x}}) = \sum_S p(S) \textbf{\textit{x}}^S$ with a multiindex $S$. We have $||S|| = 2 (n+2)$.

\begin{ex}
For $n=1$ we have $J = (j_2)$ and the coefficients $a(j_2)$ are defined by the recursion
\[
a(j_2) = {c \over 2 (1 + 2 \delta)} (4 j_2 + \delta - 3) (4 j_2 + \delta - 2) a(j_2 -1)
\]
with the initial condition a(0) = 1.
\end{ex}

\begin{ex}
For $n=2$ we have $J = (j_2, j_3)$ and the coefficients $a(j_2, j_3)$ are defined by the recursion
\[
a(J) 
= {c \over 2 (1 + 2 \delta)} (||J|| + \delta - 3) (||J|| + \delta - 2) a(J - (1, 0))
+ 2 (j_2 + 1) a(J + (1, -1))
+ 2 (j_3 + 1) p(2,0) a(J + (-2,1))
\]
with the initial conditions $a(0, 0) = 1$ and  $a(j_2, j_3) = 0$ if $j_2 < 0$ or  $j_3 < 0$.

Set 
\[
T_1 a(j_2, j_3) = a(j_2 - 1, j_3), \quad T_2 a(j_2, j_3) = a(j_2 + 1, j_3 - 1).
\]

Thus this recursion is equivalent to the discrete differential equation $W a(j_2, j_3) = 0$ for the operator
\[
W = {c \over 2 (1 + 2 \delta)} (||J|| + \delta - 3) (||J|| + \delta - 2) T_1
+ 2 (j_2 + 1) T_2
+ 2 (j_3 + 1) p(2,0) T_1 T_2^{-1} - 1
\]
with the initial conditions $\{a(0,0) = 1, a(j_2, j_3) = 0 \text{ for } \min(j_2, j_3) < 0\}$.
\end{ex}

Set 
\[
T_1 a(J) = a(j_2-1, j_3, \dots, j_{n+1}), \quad T_k a(J) = a(j_2, \dots j_k + 1, j_{k+1} - 1 \dots, j_{n+1}), \quad k = 2, \dots, n,
\]
and
\[
T_S a(J) = T_1^{s_2+s_3+ \dots + s_{n+1}-1} T_2^{s_3+ \dots + s_{n+1}-1} \dots T_n^{s_{n+1}-1} a(J) = T_1^{-1} T_2^{-1} \dots T_n^{-1} a(J - S).
\]

\begin{thm}
The coefficients $a(J)$ are defined as a solution to the discrete differential equation $W a(J) = 0$ for the operator
\[
W = {c \over 2 (1 + 2 \delta)} (||J|| + \delta - 2) (||J|| + \delta - 3) T_1 + \sum_{k = 2}^n 2 (j_{k} + 1) T_k + 2 \sum_S (j_{n+1} + 1) p(S) T_S - 1
\]
with the initial conditions $\{a(0,0) = 1, a(J) = 0 \text{ for } \min(j_k) < 0\}$.
The equation $W a(J) = 0$ expresses $a(J)$ as a linear combination of $a(J')$ with $||J'|| < ||J||$. 
\end{thm}

The proof is a direct substitution of \eqref{disform} into \eqref{diseq} using $||S|| = 2(n+2)$.

\begin{cor}
For $c \geqslant 0$ and $p(S) \geqslant 0$ for any $S$ we obtain $a(J) \geqslant 0$ for any $J$.
\end{cor}

A {\it Hurwitz series} over commutative associative ring $A$ is a formal power series in the form
 \[
\varphi(z)=\sum_{k\geqslant 0}\varphi_k \frac{z^k}{k!} \;\in\; A\otimes \mathbb{Q}[[z]] \text{ with } \varphi_k \in A.
\]

\begin{cor}
If ${c \over (1 + 2 \delta)} \in \mathbb{Z}$ and $p(S) \in \mathbb{Z}$, then $\Phi(z; \textbf{\textit{x}})$ is a Hurwitz series of $z$ over the ring $\mathbb{Z}[\textbf{\textit{x}}]$.
\end{cor}

\section{Rational solutions.} \text{}

In this section we describe a construction of differential equations that leads to examples of rational solutions of differential equations $\mathcal{D}_{P_n, n+1} = 0$.

For some function $h = h(t)$ and a constant $b$ consider
 an $(n+2) \times (n+2)$ -matrix
\[
S_n(h) = 
\begin{pmatrix}
 b h & - 1 & \dots & 0\\
 b h' & b h & \dots & 0\\
\vdots & \vdots & & \\
 {b \over n!} h^{(n)} & {b \over (n-1)!} h^{(n-1)} & \dots & - (n+1) \\
 {b \over (n+1)!} h^{(n+1)} &  {b \over n!} h^{(n)} & \dots & b h
\end{pmatrix}.
\]

The equation
\begin{equation} \label{Sn}
{1 \over b} \det S_n(h) = 0
\end{equation} 
is an ordinary differential equation homogeneous with respect
to the grading $\deg h = -4$, $\deg t = 4$.
We have
\[
{1 \over b} \det S_0(h) = h' + b h^{2} 
\]
and
\[
{1 \over b} \det S_n(h) = h^{(n+1)} + (n+2)  b h h^{(n)} + \dots + b^{n+1} h^{n+2} \qquad \text{for} \quad n>0.
\]

\begin{thm}
For any integer $n \geqslant 0$ the function 
\begin{equation} \label{32}
h(t) = h_n(t) = {1 \over b} \sum_{k=1}^{n+1} {1 \over t - a_k}
\end{equation}
is the general solution to the equation \eqref{Sn}.
\end{thm}

\begin{proof}
Consider the Newton polynomials $s_q = \sum_{k=1}^{n+1} z_k^q$ of $n+1$ variables where $z_k(t) = {1 \over (a_k - t)}$.
We have 
$
{- b \over q!} h^{(q)}(t) = s_{q+1}.
$
Using a classical result of the theory of symmetric functions we get $S_n(h) = (n+2)! \sigma_{n+2}(z_1, \dots, z_{n+1}, 0) \equiv 0$, where $\sigma_{n+2}()$ is the elementary symmetric function. 
Thus the function \eqref{32} is the solution to the differential equation \eqref{Sn} of order $(n+1)$. It defines the general solution to this equation because it depends of $(n+1)$ parameters $a_1$, $\dots$, $a_{n+1}$. 
\end{proof}

\begin{cor}
For each positive integer $n$ for the existence of a polynomial $P_n$ such that 
\[
\mathcal{D}_{P_n, n+1}(h) = {1 \over b} \det S_n(h)
\]
it is necessary that
\[
b = (n+1).
\]
\end{cor}

The calculations for $n=0,1,2,3,4$ show that this condition is sufficient (see below).

\begin{problem}
Prove that in the notions of the corollary the condition $b = (n+1)$ is sufficient.

\end{problem}

\section{Special cases.} \text{}

In this section we construct explicitly in the cases $n = 0, 1,2,3,4$ the formulas we have derived above. 

\subsection{Case $n = 0$.}
We have $\mathcal{D}_1(h) = h' + h^2$.
 In this case there are no variables $x_k$ and $P_0 = 0$.
The system \eqref{e14} has the form
\[
{d \over d t} r = - (\delta + {1 \over 2}) h, \qquad {d \over d t} h = - h^2.
\]
The general solution of this system is
\[
h(t) = {a \over (a t - b)}, \quad r(t) = - ({1 \over 2}  + \delta) \ln(a t - b) + r_0,
\]
where the parameters $r_0 \in \mathbb{C}$ and $(a:b) \in \mathbb{C}P^1$ are constant. Thus the function
\[
\psi(z,t) = {\exp{r_0} \over (a t - b)^{{1 \over 2} + \delta}} \exp\left(- {a z^2 \over 2 (a t - b)}\right) z^\delta
\]
solves the heat equation. In the case $\delta = 0$, $r_0 = 0$, $b/a = c$ it coincides with the solution \eqref{fexp} of the heat equation.
\begin{cor}
For $\delta = 0$ and $b/a < t < \infty$ the function $\psi(z,t)$ is the Gaussian density distribution.
\end{cor}

\subsection{Case $n = 1$.}
We have
\[
\mathcal{D}_2(h) = ({d \over d t} + 4 h) \mathcal{D}_{1}(h)
\]
and $\deg x_2 = - 8$, $\deg P_1 = - 12$, thus $P_1(x_2) \equiv 0$.
The system \eqref{e14} has the form
\[
{d \over d t} r = - (\delta + {1 \over 2}) h, \qquad
{d \over d t} h = - h^2 + x_2, \qquad
{d \over d t} x_{2} = - 4 h x_{2}.
\]

Therefore $h(t)$ is a solution to
\begin{equation} \label{odd1}
h'' + 6 h h' + 4 h^3 = 0.
\end{equation}

\begin{rem} Differentiate \eqref{odd1} and put $y(t) = 2 h(t)$ 
to get Chazy-4 equation:
\[
y''' = - 3 y y'' - 3 y'^2 - 3 y^2 y'.
\]
\end{rem}

The general solution to \eqref{odd1} 
has the form
\begin{equation} \label{h24}
h(t) = {1 \over 2} \left({1 \over t-a} + {1 \over t-b}\right).
\end{equation}

We have $x_2(t) = - {1 \over 4} \left({1 \over(t-a)} - {1 \over(t-b)}\right)^2$.

In the case $a=b$ we come to the case $n=0$.

According to the key theorem, we have
\begin{equation} \label{Phz}
\Phi(z; x_2) = z^\delta + \sum_{q = 1}^\infty \Phi_{2q}(x_2) {z^{4 q+\delta} \over (4 q + \delta)!},
\end{equation}
where
\[
\Phi_2 = - 2 (1 + 2 \delta) x_2, \quad \Phi_{2 q} = - (4 q + \delta - 3) (4 q + \delta - 2) x_2 \Phi_{2q-2}, \; \; q = 2, 3 \dots.
\]

Note the function $\Phi(z; x_2)$ solves
\begin{equation} \label{n1eq}
{d^2 \over d z^2} \Phi(z; x_2) = - x_2 z^2 \Phi(z; x_2).
\end{equation}

\begin{lem} $\Phi(z; x_2) = z^\delta \gamma(z^4; x_2, \delta)$, 
where $\gamma(v; x_2, \delta)$ is a solution to the differential equation
\[
\gamma'(v) + {4 v \over 3 + 2 \delta} \gamma''(v) = \lambda \gamma(v), \qquad \text{with} \quad \lambda = - {1 \over 4 (3 + 2 \delta)} x_2
\]
with the initial condition $\gamma(0) = 1$.
\end{lem}

Thus the function $\gamma(v)$ is the eigenfunction of the generalized
shift operator, defined by the generator
\[
d = {d \over d v} + {4 \over 3 + 2 \delta} v {d^2 \over d v^2}.
\]

\subsection{Case $n = 2$.}
This case leads to remarkable equations with Painlev\'e properties (see, e.g., \cite{CM}).

We have $P_2(x_2, x_3) = c_4 x_2^2$, where $c_4$ is a constant and
\[
\mathcal{D}_3(h) = ({d \over d t} + 6 h) \mathcal{D}_{2}(h), \qquad \mathcal{D}_{P_2,3}(h) = \mathcal{D}_{3}(h) -  c_4 \mathcal{D}_1(h)^2.
\]

The system \eqref{e14} has the form
\[
{d \over d t} r = - (\delta + {1 \over 2}) h, \quad 
{d \over d t} h = - h^2 + x_2,
 \quad
{d \over d t} x_2 = x_{3} - 4 h x_2, \quad 
{d \over d t} x_{3} = c_4 x_2^2 - 6 h x_{3}.
\]

Therefore $h(t)$ is a solution to
\begin{equation} \label{odd2}
h''' + 12 h h'' - 18 (h')^2 + ( 24 - c_4) ( h' + h^2)^2 = 0.
\end{equation}
Equation \eqref{odd2} is brought by the substitution $y(t) = - 6 h(t)$ to
\begin{equation} \label{Chn}
y''' = 2 y y'' - 3 (y')^2 + {24 - c_4 \over 216} (6 y' - y^2)^2.
\end{equation}

For ${24 - c_4 \over 216} = - {4 \over k^2 - 36}$ the equation \eqref{Chn} is Chazy-12.
For $c_4 = 24$ this equation is Chazy-3. 

 For $c_4 = 6$ the equation \eqref{Chn} 
takes the form
\begin{equation} \label{odd2v}
y''' = 2 y y'' - y^2 y' + {1 \over 12} y^4.
\end{equation}
The equation \eqref{odd2v} is linear with respect to derivatives.

We have
\[
{1 \over b} \det S_2(h) =  h''' + 4 b h h'' + 3 b (h')^2 + 6 b^2 h^2 h' + b^3 h^4.
\]
For $b= 3$ we obtain the equation \eqref{odd2} with $c_4 = -3$.

\begin{cor} 
The function 
\[
 - 2 \left({1 \over t-a_1} + {1 \over t-a_2} + {1 \over t - a_3} \right)
 \] is the general solution of the equation Chazy-12 (see \eqref{Chn})
  with $c_4 = -3$ ($k^2= 4$).
\end{cor}

We will describe solutions of the heat equation in the case
 when $y(t) = - 6 h(t)$ is a solution to Chazy-3. 
Such solutions may be constructed explicitly in terms of the Weierstrass sigma-function.

\subsection{Necessary facts from the theory of elliptic functions.}

Details of the basic facts from the theory of elliptic functions see e.g. \cite{WW}.

Consider the elliptic curve in standard Weierstrass form
\[
V = \{ (\lambda, \mu) \in \mathbb{C}^2 : \mu^2 = 4 \lambda^3 - g_2 \lambda - g_3 \}.
\]
It is non-degenerate for $g_2^3 \ne 27 g_3^2$. Set
\[
2 \omega_k = \oint_{a_k} {d\lambda \over \mu}, \qquad 2 \eta_k = - \oint_{a_k} {\lambda d\lambda \over \mu}, \qquad k = 1,2,
\]
where ${d\lambda \over \mu}$ and ${\lambda d\lambda \over \mu}$ are basis holomorphic differentials 
and $a_k$ are basis cycles on the curve such that
\[
\eta_1 \omega_2 - \omega_1 \eta_2 = {\pi i \over 2}.
\]

A plane non-degenerate algebraic curve~$V$ 
defines the rank $2$ lattice $\Gamma \subset \mathbb{C}$ 
generated by $2 \omega_1$ and $2 \omega_2$ with ${\bf Im} {\omega_2 \over \omega_1} > 0$. 

The {\it Jacobian} of the curve $V$ is the complex torus $\mathbb{T} = \mathbb{C}/\Gamma$.

An {\bf elliptic function} is a meromorphic function on $\mathbb{C}$ such that
\[
f(z + 2 \omega_1) = f(z), \quad f(z + 2 \omega_2) = f(z).
\]

The {\bf Weierstrass function $\wp(z; g_2, g_3)$} is the unique elliptic function with periods $2 \omega_1$, $2 \omega_2$ and poles only in lattice points such that 
\[ 
\underset{z\to 0}{\lim}\left( \wp(z)-\frac{1}{z^2}
\right)=0.
\]
It defines the uniformization of the elliptic curve 
in the standard Weierstrass form
\[
\wp'(z)^2 = 4 \wp(z)^3 - g_2 \wp(z) - g_3.
\]

{\bf The Weierstrass $\sigma$-function}
is the entire {\bf odd} function $\sigma(z) = \sigma(z; g_2, g_3)$ such that 
\[\big(\ln \sigma(z; g_2, g_3)\big)'' =  - \wp(z; g_2, g_3) \qquad
\text{and}  \qquad \underset{z\to 0}{\lim}\left( \frac{\sigma(z)}{z} \right)=1.\]

Periodic properties:
\[
\sigma(z+2\omega_k) = - \sigma(z) \exp\big(2\eta_k(z+\omega_k)\big), \quad k = 1,2.
\]

Degenerate case:
\[
\sigma\left(z;  {4 \over 3} a^4, {8 \over 27} a^6\right) = {1 \over a} \exp\left({1 \over 6} a^2 z^2\right) \sin{a z} .
\]

Consider the fields on $\mathbb{C}^2$ 
\[
l_0 = 4 g_2 {\partial \over \partial g_2} + 6 g_3 {\partial \over \partial g_3}, \qquad
l_2 = 6 g_3 {\partial \over \partial g_2} + {1 \over 3} g_2^2 {\partial \over \partial g_3}.
\]
We have $[l_0, l_2] = 2 l_2$, \; $l_0 \Delta = 12 \Delta$, \; $l_2 \Delta = 0$, \; $\langle l_0, l_2\rangle = {4 \over 3} \Delta$.

\begin{thm}[The Weierstrass theorem]
The operators 
\[
Q_0 = z {\partial \over \partial z} - 1 - l_0, \quad Q_2 = {1 \over 2} {\partial^2 \over \partial z^2} + {1 \over 24} g_2 z^2 - l_2
\]
annihilate the sigma-function, that is
\[
Q_0 \sigma(z; g_2, g_3) = 0, \qquad Q_2 \sigma(z; g_2, g_3) = 0.
\]
\end{thm}

\begin{thm}[See \cite{Trudy}]
The function $\psi(z, t)$ such that
\begin{equation}\label{vid}
\psi(z, t) = e^{- {1 \over 2} h(t) z^2 + r(t)} \sigma\left(z, g_2(t), g_3(t)\right)
\end{equation}
for some functions $r(t)$, $h(t)$, $g_2(t)$ and $g_3(t)$  satisfies the heat equation
\begin{equation}\label{F-11}
{\partial \over \partial t} \psi(z, t) = {1 \over 2}
{\partial^2 \over \partial z^2} \psi(z, t)
\end{equation}
if and only if the functions $r(t)$, $h(t)$, $g_2(t)$ and $g_3(t)$ satisfy the homogeneous polynomial dynamical system
in $\mathbb{C}^4$ 
with coordinates $(r, h, g_2, g_3)$, $\deg h = -4$, $\deg r = 0$:
\begin{equation} \label{ds}
r' = - {3 \over 2} h, \quad
h' = - h^2 + {1 \over 12} g_2, \quad 
g_2' = 6 g_3 - 4 h g_2, \quad 
g_3' = {1 \over 3} g_2^2 - 6 h g_3. 
\end{equation}
\end{thm}

\begin{ex}
The function
\[
\psi(z,t) = \exp\left(- {1 \over 2} a^2 t\right) {\sin a z \over a}, \quad a = const
\]
is a {\bf periodic odd} function of $z$ with the initial conditions $\psi(0,t) = 0$,
$\psi'(0,t) = \exp\left(- {1 \over 2} a^2 t\right)$.
It is a classical solution of the heat equation.
In this case we have
\[
r = - {1 \over 2} a^2 t, \quad h = {1 \over 3} a^2, \quad g_2 = {4 \over 3} \gamma^4, \quad g_3 = {8 \over 27} \gamma^6.
\]
\end{ex}

\begin{ex}
For the classical solution 
\[
\psi(z,t) = \psi_*(z,t) - \psi_*(- z,t), \quad
\text{where} \quad
\psi_*(z,t) =  {1 \over \sqrt{t}} \exp{\left(- {(z-a)^2 \over 2 t}\right)},
\]
which is decreasing when $z \to \pm \infty$, we have $\gamma = - {i a \over t}$ and
\[ h = - {a^2 - 3 t \over 3 t^2}, \quad r = \ln({2 a \over \sqrt{t^{3}}}) - {a^2 \over 2 t}, \quad g_2 = {4 \over 3} {a^4 \over t^4}, \quad g_3 = - {8 \over 27} {a^6 \over t^6}.
\]
For this solution we have $\psi(0,t) = 0$ and 
$\psi'(z,t) = {2 a \over t \sqrt{t}} \exp{\left(- {a^2 \over 2 t}\right)}.$
\end{ex}

Consider the dynamical system \eqref{ds}.
It corresponds to the system generated by the field $l_2$:
\begin{equation} \label{systeml2}
{d \over d \tau} g_2 = 6 g_3, \qquad {d \over d \tau} g_3 = {1 \over 3} g_2^2.
\end{equation}
For $x_2 = {1 \over 12} g_2$, $x_3 = {1 \over 2} g_3$ we get the case $n=2$, $\delta = 1$, $c_4 = 24$.

In this case the equation \eqref{odd2} takes the form
\[
h''' + 12 h h'' - 18 (h')^2 = 0
\]
and is brought by the substitution $y(t) = - 6 h(t)$ 
to the Chazy-3 equation
\[
y''' = 2 y y'' - 3 (y')^2.
\]

The system \eqref{systeml2} has the solution
\begin{equation} \label{b20}
g_2(\tau) = 3 \wp(\tau + d; 0, b_3), \qquad g_3(\tau) = {1 \over 2} \wp'(\tau + d; 0, b_3),
\end{equation}
where $b_3 = {4 \over 27} g_2(0)^3 - 4 g_3(0)^2$, and $d$ is the solution of 
the compatible system $\wp(d; 0, b_3) = {1 \over 3} g_2(0)$, $\wp'(d; 0, b_3) = 2 g_3(0)$.

In the terms of our construction,
$x_2 = {1 \over 4} \wp(\tau + d; 0, b_3)$ for $n = 2$, $c_4 = 24$. 
Further we will show that for $n = 3$ solutions to the heat equation with $x_2 = {1 \over 4} \wp(\tau + d; b_2, b_3)$ appear. 

The corresponding function $\Phi(z;\textbf{\textit{x}})$ can be considered as a series of the 
parameter $b_2$ and at $b_2 = 0$ it coincides 
with the Weierstrass sigma-function (see case $n=2$).

\subsection{Case $n = 3$.}
\[
\mathcal{D}_4(h) = \left({d \over d t} + 8 h\right) \mathcal{D}_3(h).
\]
We have $P_3(x_2, x_3, x_4) = c_5 x_2 x_3$, where $c_5$ is a constant and
\[
\mathcal{D}_{P_3,4}(h) = \mathcal{D}_4(h) - c_5 \mathcal{D}_1(h) \mathcal{D}_2(h).
\]

The system \eqref{e14} has the form
\begin{multline*}
{d \over d t} r =  - (\delta + {1 \over 2}) h,
\quad {d \over d t} h = - h^2 + x_2, \quad
{d \over d t} x_2 = x_3 - 4 h x_2, \quad 
{d \over d t} x_3 = x_4 - 6 h x_3, \quad 
{d \over d t} x_4 = c_5 x_2 x_3 - 8 h x_4.
\end{multline*}

Therefore $h(t)$ is a solution to
\begin{equation} \label{an3}
h'''' + 20 h h''' - 24 h' h'' + 96 h^2 h'' - 144 h (h')^2 + (48 - c_5) ( h' + h^2) ( h'' + 6 h h' + 4 h^3) = 0.
\end{equation}

\begin{lem} For $c_5 = 2 c_4$ the formula holds
\[
\mathcal{D}_{P_3,4}(h) = \left({d \over d t} + 8 h\right) \mathcal{D}_{P_2,3}(h).
\]
\end{lem}

Thus for $c_5 = 48$ the equation \eqref{an3} is obtained 
from the Chazy-3 equation by the substitution $y(t) = - 6 h(t)$ 
and the action of the operator $\left({d \over d t} + 8 h\right)$.

\begin{lem} For $c_5 = 24$ and $y(t) = - 2 h(t)$ the equation \eqref{an3} 
takes the form
\begin{equation} \label{an35}
y'''' - 10 y y''' + 30 y^2 y'' - 30 y^3 y' + 6 y^5 = 0.
\end{equation} 
The equation \eqref{an35} is linear with respect to derivatives.
\end{lem}

We have
\[
{1 \over b} \det S_3(h) = h'''' + 5 b h h''' + 10 b h' h'' + 10 b^2 h^2 h'' + 15 b^2 h (h')^2 + 10 b^3 h^3 h' + b^4 h^5.
\]
For $b = 4$ we obtain the equation \eqref{an3} with $c_5 = - 16$.

\begin{cor} 
The function 
\[
{1 \over 4} \left({1 \over t-a_1} + {1 \over t-a_2} + {1 \over t - a_3} + {1 \over t - a_4} \right)
 \] is the general solution of the equation \eqref{an3} with $c_5 = - 16$.
\end{cor}

Consider the functions \eqref{b20} (case $n=2$) as the limit of functions
\begin{equation} \label{gb23}
g_2(\tau) = 3 \wp(\tau + d; b_2, b_3), \quad g_3(\tau) = {1 \over 2} \wp'(\tau + d; b_2, b_3)
\end{equation}
at $b_2 \to 0$. The functions \eqref{gb23} along with $g_4 = {- b_2 \over 8}$ give 
the general solution to the system
\[
{d \over d \tau} g_2 = 6 g_3, \qquad {d \over d \tau} g_3 = {1 \over 3} g_2^2 + 2 g_4, \qquad {d \over d \tau} g_4 = 0. 
\]
For $x_2 = {1 \over 12} g_2$, $x_3 = {1 \over 2} g_3$, $x_4 = {1 \over 6} g_2^2 + g_4$ this system
\[
{d \over d \tau} x_2 = x_3, \qquad {d \over d \tau} x_3 = x_4, \qquad {d \over d \tau} x_4 = - 96 x_2 x_3, 
\]
corresponds to the case $n=3$, $c_5 = 48$. 
 
Therefore the solution $\psi(z,t)$ of the heat equation corresponding to this case has a parameter
$b_2$, and for $b_2 \to 0$ it tends to 
the solution corresponding to the case $n = 2$,  $c_4 = 24$.

Thus the function
$
e^{{1 \over 2} h(t) z^2 - r(t)} \psi(z;t)
$ at $t = 0$ 
is a deformation of the Weierstrass sigma-function $\sigma(z; g_2, g_3)$ with the deformation parameter $b_2$.

\subsection{Case $n = 4$.}
\[
\mathcal{D}_5(h) = \left({d \over d t} + 10 h \right) \mathcal{D}_4(h).
\]
We have $P_4(x_2, x_3, x_4, x_5) = c_{6,2} x_2^3 + c_{6,3} x_3^2 +  c_{6,4} x_2 x_4$, 
where $c_{6,2}$, $c_{6,3}$ and $c_{6,4}$ are constants,
\[
\mathcal{D}_{P_4,5}(h) = \mathcal{D}_5(h) - c_{6,2} \mathcal{D}_1(h)^3 - c_{6,3} \mathcal{D}_2(h)^2 - c_{6,4} \mathcal{D}_1(h) \mathcal{D}_3(h).
\]

The system \eqref{e14} has the form
\begin{multline*} 
{d \over d t} r = - (\delta + {1 \over 2}) h, \quad
{d \over d t} h = - h^2 + x_2, \quad
{d \over d t} x_2 = x_3 - 4 h x_2, \quad
{d \over d t} x_3 = x_4 - 6 h x_3, \\
{d \over d t} x_4 = x_5 - 8 h x_4, \quad
{d \over d t} x_5 = c_{6,2} x_2^3 + c_{6,3} x_3^2 + c_{6,4} x_2 x_4 - 10 h x_5.
\end{multline*} 

\begin{lem}
For $c_{6,2} = 0$, $c_{6,3} = c_5$ and $c_{6,4} = c_5$ the formula holds
\[
\mathcal{D}_{P_4,5}(h) = \left({d \over d t} + 10 h\right) \mathcal{D}_{P_3,4}(h).
\]
\end{lem}

Thus for $c_{6,2} = 0$, $c_{6,3} = 48$ and $c_{6,4} = 48$ 
the equation $\mathcal{D}_{P_4,5}(h) = 0$ is obtained
from the Chazy-3 equation by the substitution $y(t) = - 6 h(t)$ and the action of the operator
\[
\left({d \over d t} + 10 h\right) \left({d \over d t} + 8 h\right).
\]

\begin{lem} For $c_{6,2} = - 120$, $c_{6,3} = 24$, $c_{6,4} = 44$ 
and $y(t) = - 10 h(t)$ the equation $\mathcal{D}_{P_4,5}(h) = 0$ takes the form
\begin{equation} \label{an4}
 y''''' - 3 y y'''' + 3 y^2 y''' - {6 \over 5} y^3 y'' + {9 \over 50} y^4 y' - {3 \over 500} y^6 = 0.
\end{equation} 
The equation \eqref{an4} is linear with respect to derivatives.
\end{lem}

We have
\begin{multline*}
\det S_4(h) = h''''' + 6 b h h'''' + 15 b (h' +  b h^2) h''' + 10 b (h'' + 6 b h h' + 2 b^2 h^3) h'' + \\ + 15 b^2 \left((h')^2 + 3 b h^2 h' + b^2 h^4\right) h' + b^5 h^6.
\end{multline*}
For $b = 5$ we obtain $\mathcal{D}_{P_4,5}(h)$ 
with $c_{6,2} = - 45$, $c_{6,3} = -26$ and $c_{6,4} = -31$.

\begin{cor}
The function 
\[
{1 \over 5} \left({1 \over t-a_1} + {1 \over t-a_2} + {1 \over t - a_3} + {1 \over t - a_4} + {1 \over t - a_5} \right)
 \]
 is the general solution of the equation $\mathcal{D}_{P_4,5}(h) = 0$ 
  with $c_{6,2} = - 45$, $c_{6,3} = -26$ and $c_{6,4} = -31$.
\end{cor}

\section{Addendum.} \text{}

In the main part of the paper we have examined the ansatz \eqref{sol} for solutions of the heat equation \eqref{heat}. Remark that ansatz \eqref{sol} is a  special case of 
\begin{equation} \label{sol2}
\psi(z,t) = e^{r(t)} \Psi(z; \textbf{\textit{x}}(t))
\end{equation}
where
\begin{equation} \label{e52}
\Psi(z; \textbf{\textit{x}}) = z^\delta + \sum_{k \geqslant 1} \Psi_k(\textbf{\textit{x}}) {z^{2 k+ \delta} \over (2 k + \delta)!},
\end{equation}
$\deg z = 2$, $\textbf{\textit{x}} = (x_1, x_2, \dots x_{n+1})$, $\deg x_q = - 4 q$,
and $\Psi_k(\textbf{\textit{x}})$ are homogeneous polynomials of degree $-4 k$.

Using the approach developed in sections 3 and 4, we come to the following results:

\begin{thm} \label{t52}
Among the following conditions each two imply the third one:

1) The function $\psi(z,t)$ in the form \eqref{sol2}
solves the heat equation.

2) The coefficients of $\Psi(z; \textbf{\textit{x}})$ are defined by the recursion 
\[
\Psi_{k+1}(\textbf{\textit{x}}) = 2 \sum_{j=1}^{n+1} p_{j+1}(\textbf{\textit{x}}) {\partial \over \partial x_j} \Psi_k(\textbf{\textit{x}}) + \Psi_1(\textbf{\textit{x}}) \Psi_k(\textbf{\textit{x}})
\]
and the condition
\[
r'(t) = {1 \over 2} \Psi_1(\textbf{\textit{x}}(t)).
\]

3) The functions $x_k(t)$ satisfy the homogeneous polynomial dynamical system 
\begin{equation} \label{ds2}
{d \over d t} x_k = p_{k+1}(\textbf{\textit{x}}), \qquad k = 1, 2, \dots, n+1.
\end{equation}
\end{thm}

\begin{ex}
For $n = 2$ let $x_1(t) = {1 \over 3} \left({1 \over t - a_1} + {1 \over t - a_2} + {1 \over t - a_3} \right)$. Then for $x_2(t) = - {1 \over 3} \left({1 \over (t - a_1)^2} + {1 \over (t - a_2)^2} + {1 \over (t - a_3)^2} \right)$ and $x_3(t) = {2 \over 3} \left({1 \over (t - a_1)^3} + {1 \over (t - a_2)^3} + {1 \over (t - a_3)^3} \right)$ we have
\[
{d \over d t} x_1 = x_2, \quad {d \over d t} x_2 = x_3, \quad {d \over d t} x_3 = - 3 (4 x_1 x_3 + 3 x_2^2 + 18 x_2 x_1^2 + 9 x_1^4).
\]
Denote $\mathcal{L} = 2 x_2 {\partial \over \partial x_1} + 2 x_3 {\partial \over \partial x_2} - 6 (4 x_1 x_3 + 3 x_2^2 + 18 x_2 x_1^2 + 9 x_1^4) {\partial \over \partial x_3} + \Psi_1(\textbf{\textit{x}})$.

Theorem \ref{t52} implies that for the function $\Psi(z; \textbf{\textit{x}})$ of the form \eqref{e52} where $\Psi_1(\textbf{\textit{x}}) = - {1 \over 2} x_1$ and $\Psi_k(\textbf{\textit{x}}) = \mathcal{L}^{k-1} \Psi_1(\textbf{\textit{x}})$
the function
\[
\psi(z,t) = {1 \over \sqrt[12]{(t-a_1)(t-a_2)(t-a_3)}} \Psi(z; \textbf{\textit{x}}(t))
\]
satisfies the heat equation.
\end{ex}

Using the methods of section 4, we bring the system \eqref{ds2} to the form
\[
{d \over d t} x_k = x_{k+1}, \qquad k = 1, 2, \dots, n, \qquad {d \over d t} x_{n+1} = K_n(\textbf{\textit{x}}).
\]

\begin{thm} \label{92}
Each solution of the heat equation of the form \eqref{sol2} is defined by the set $(n, K_n, x_1, r_0)$, where
$n$ is a natural number, 
$K_n$ is a homogeneous polynomial $K_n(\textbf{\textit{x}})$ of degree $- 4 (n+2)$, 
$x_1$ is a solution $x_1(t)$ of the equation $x_1^{(n+1)} = K_n(x_1, x_1', \dots, x_1^{(n)})$ and
$r_0$ is a constant.
\end{thm}

\begin{cor}
Each solution of the heat equation of the form \eqref{sol2} is defined by a finite-dimensional numerical vector $(n, \textbf{\textit{K}}_n, \textbf{\textit{C}}_n, r_0)$, where $ \textbf{\textit{K}}_n$ is the vector of coefficients of the polynomial $K_n(\textbf{\textit{x}})$ and $\textbf{\textit{C}}_n$ is the vector of initial data in the Cauchy problem for the equation $x_1^{(n+1)} = K_n(x_1, x_1', \dots, x_1^{(n)})$.
\end{cor}

The linear space of homogeneous polynomials $K_n(\textbf{\textit{x}})$ of degree $- 4 (n+2)$ is $p(n+2) - 1$ dimensional, 
where $p(n)$ is the number of partitions of $n$ (compare to lemma \ref{l1}).

\begin{cor}
For each $n$ the space of solutions of the heat equation of the form \eqref{sol2} is $p(n+2)+n+1$-dimensional and has the $p(n+2) - p(n+1) + n + 1$-dimensional subspace of solutions of the form \eqref{sol}.
\end{cor}
 
Therefore the introduction of the solution of the heat equation ansatz described in the main part of the work picks out a special class of non-linear differential equations of the form $x_1^{(n+1)} = K_n(x_1, x_1', \dots, x_1^{(n)})$ (see section \ref{c34}).

\end{document}